\def\jolabels{0}
\def\numcolumns{1}

\if\numcolumns2
\documentclass[10pt,conference,twocolumn]{IEEEtran}
\else \if\numcolumns1
\documentclass[11pt,journal,onecolumn]{IEEEtran}
\renewcommand{\baselinestretch}{1.48}
\else 
 \typein{*** Press Enter ***} \stop \fi \fi

\renewcommand{\QED}{\QEDopen}

\usepackage{epsfig}
\usepackage{myamsthm}
\theoremstyle{plain}

\usepackage{color}
\usepackage{jolabels}

\newtheorem{lemma}{Lemma}
\newtheorem{theorem}{Theorem}

\newtheorem{definition}{Definition}

\newtheorem{conjecture}{Conjecture}

\newcommand{\rem}[1]{}

\newcommand{\figref}[1]{Fig.~\ref{fig:#1}}

\newcommand{\thrmref}[1]{Theorem~\ref{thm:#1}}
\newcommand{\lemref}[1]{Lemma~\ref{lem:#1}}

\long\def\symbolfootnote[#1]#2{\begingroup%
\def\thefootnote{\fnsymbol{footnote}}\footnote[#1]{#2}\endgroup}

\newcommand{\beq}{\begin{equation}}
\newcommand{\eeq}{\end{equation}}

\usepackage{amsmath}
\usepackage{fancybox}
\usepackage{latexsym}

\usepackage{graphicx}
\usepackage{amsfonts}
\usepackage{amssymb}

\newfont{\boldlarge}{msbm10 scaled 1100}

\newlength{\myVSpace}% the height of the box
\renewcommand{\theenumi}{{\bf\alph{enumi}}}
\makeatother

\begin{document}

% paper title
\title{\bf The Rate Loss of Single-Letter Characterization: \\
The ``Dirty'' Multiple Access Channel}

%Parts of this work was
%presented at ISIT2006 Adelaide, Australia September 2005
%\footnotemark[2]\footnotetext[2]{Bla Bla

% author names and affiliations

\author{
\large Tal Philosof and Ram Zamir \symbolfootnote[2]{This research
was partially supported by BSF grant No-2004398}
\\
%\normalsize Dept. of Elect. Eng. - Systems, Tel Aviv University,
%ISRAEL\\
\normalsize Dept. of Electrical Engineering - Systems, Tel-Aviv
University\\
\normalsize Tel-Aviv 69978, ISRAEL\\
%\\
%{\em Submitted to the IEEE-Israel Convention, Dec. 2002}
{\em talp,zamir@eng.tau.ac.il}\\
\normalsize \textit{Submitted to IEEE Trans. on Information Theory March 2008}\\
 }

% make the title area
\maketitle

%%%%%%%%%%%%%%%%%%%%%%%%%%%%%%%%%%%%%%%%%%%%%%%%%%%%%%%%%%%%%%%%%%%%%%%%%%%
\begin{abstract}
% Information theory considers a problem ``solved'' if a coding
% theorem, consisting of a direct part and a converse part that meet
% each other, is proved.
For general memoryless systems, the typical information theoretic
solution - when exists - has a ``single-letter'' form. This
reflects the fact that optimum performance can be approached by a
random code (or a random binning scheme), generated using
independent and identically distributed copies of some
single-letter distribution. Is that the form of the solution of
any (information theoretic) problem? In fact, some counter
examples are known. The most famous is the ``two help one''
problem: Korner and Marton showed that if we want to decode the
modulo-two sum of two binary sources from their independent
encodings, then linear coding is better than random coding.
In this paper we provide another counter example, the
``doubly-dirty'' multiple access channel (MAC). Like the
Korner-Marton problem, this is a multi-terminal scenario where
side information is distributed among several terminals; each
transmitter knows part of the channel interference but the
receiver is not aware of any part of it. We give an explicit
solution for the capacity region of a binary version of the
doubly-dirty MAC, demonstrate how the capacity region can be
approached using a linear coding scheme, and prove that the ``best
known single-letter region'' is strictly contained in it. We also
state a conjecture regarding a similar rate loss of single letter
characterization in the Gaussian case.
\end{abstract}

%%%%%%%%%%%%%%%%%%%%%%%%%%%%%%%%%%%%%%%%%%%%%%%%%%%%%%%%%%%%%%%%%%%%%%%%%%%
\begin{keywords}
Multi-user information theory, random binning, linear lattice
binning, dirty paper coding, lattice strategies, Korner-Marton
problem.
\end{keywords}
%%%%%%%%%%%%%%%%%%%%%%%%%%%%%%%%%%%%%%%%%%%%%%%%%%%%%%%%%%%%%%%%%%%%%%%%%%%

\section{Introduction}\label{sec:intro}

Consider the two-user / double-state memoryless multiple access
channel (MAC) with transition and state probability distributions
\begin{eqnarray}
P(y|x_1,x_2,s_1,s_2) \ \  \mbox{and} \ \ P(s_1,s_2),
\label{eq:GeneralModel}
\end{eqnarray}
respectively,
where the states $S_1$ and $S_2$ are known non-causally to user
$1$ and user $2$, respectively. A special case of
(\ref{eq:GeneralModel}) is the additive channel shown in
\figref{DoublyDirtyMACModel}. In this channel, called the {\em
doubly-dirty MAC} (after Costa's ``writing on dirty paper''
\cite{Costa83}), the total channel noise consists of three
independent components: $S_1$ and $S_2$, the interference signals,
that are known to user $1$ and user $2$, respectively, and $Z$,
the unknown noise, which is known to neither. The channel inputs
$X_1$ and $X_2$ may be subject to some average cost constraint.

Neither the capacity region of \eqref{eq:GeneralModel} nor that of
the special case of \figref{DoublyDirtyMACModel} are known. In
this paper we consider a particular binary version of the
doubly-dirty MAC of \figref{DoublyDirtyMACModel}, where all
variables are in $\mathbb{Z}_2$, i.e., $\{0,1\}$, and the unknown
noise $Z=0$. The channel output of the binary doubly-dirty MAC is
given by
\begin{eqnarray}
Y=X_1\oplus X_2\oplus S_1\oplus S_2\label{eq:BinaryModel},
\end{eqnarray}
where $\oplus$ denotes the $\mbox{mod}\;2$ addition (xor), and
$S_1,S_2$ are $\mbox{Bernoulli}(1/2)$ and independent. Each of the
codewords $\mathbf{x}_i\in\mathbb{Z}_2^n$  is a function of the
message $W_i$ and the interference vector
$\mathbf{s}_i\in\mathbb{Z}_2^n$, and must satisfy the input
constraint, $\frac{1}{n}w_H(\mathbf{x}_i)\leq q_i$, $\;i=1,2$,
where $0 \leq q_1,q_2\leq 1/2$ and $w_H(\cdot)$ is the Hamming
weight. The coding rates $R_1$ and $R_2$ of the two users are
given as usual by $R_i = \frac{1}{n} \log| {\cal W}_i |$, where
${\cal W}_i$ is the set of messages of user $i$, and $n$ is the
length of the codeword.
\begin{figure}[t]
%\vspace{+0.5cm}
  \begin{center}
    \input{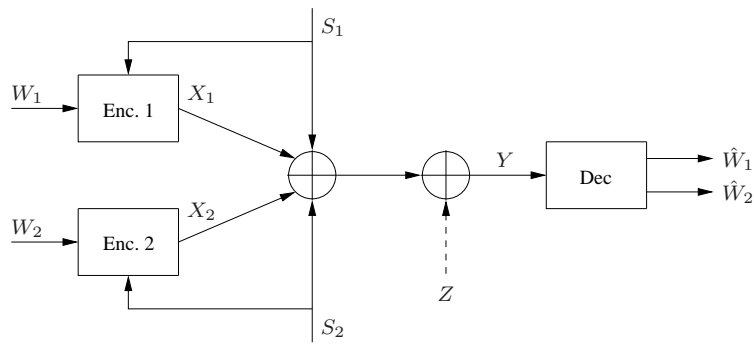}
%    \vspace{-0.2cm}
    \caption{Doubly-dirty MAC}\label{fig:DoublyDirtyMACModel}
  \end{center}
%\vspace{-0.2cm}
\label{fig1}
\end{figure}

The double state MAC \eqref{eq:GeneralModel} generalizes the point
to point channel with side information (SI) at the transmitter
considered by Gel'fand and Pinsker \cite{GelfandPinsker80}. They
prove their direct coding theorem using the framework of random
binning, which is widely used in the analysis of multi-terminal
source and channel coding problems \cite{CoverBook}. They obtain a
general capacity expression which involves an auxiliary random
variable $U$:
\begin{eqnarray}\label{eq:GelfandPinskerCapacity}
C=\max_{P(u,x|s)}\left\{H(U|S)-H(U|Y)\right\}
\end{eqnarray}
where the maximization is over all the joint distributions of the
form $p(u,s,y,x)=p(s)p(u,x|s)p(y|x,s)$.

The channel in \eqref{eq:GeneralModel} with only one informed
encoder (i.e., where $S_2=\{\emptyset\}$) was considered recently
by Somekh-Baruch et al. \cite{BaruchShamaiVerdu06} and Kotagiri
and Laneman \cite{KotagiriLaneman06}. The common message
($W_1=W_2$) capacity of this channel is known
\cite{BaruchShamaiVerdu06}, and it involves using random binning
by the informed user. For the binary ``one dirty user'' case
(i.e., (\ref{eq:BinaryModel}) with $S_2=0$), we show that
Somekh-Baruch's common-message capacity becomes (see
Appendix~\ref{App:Appendix1})
\begin{align}\label{eq:SingleInterferenceOuterBound}
C_{com}=H_b(q_1),
\end{align}
where $H_b(x)\triangleq -x\log_2(x)-(1-x)\log_2(1-x)$ is the
binary entropy function. Clearly, the doubly-dirty
individual-message case is harder. Thus, it follows from
\eqref{eq:SingleInterferenceOuterBound} that the rate-sum in the
setting of \figref{DoublyDirtyMACModel} is upper bounded by
\begin{align}
&R_1+R_2 \leq
\min\Big\{H_b(q_1),H_b(q_2)\Big\}.\label{eq:DoublyInterferenceOuterBound}
\end{align}
In \thrmref{Capacity} we show that this upper bound is in fact
tight.

One approach to find {\em achievable rates} for the doubly-dirty MAC, is
to extend the Gel'fand and Pinsker solution
\cite{GelfandPinsker80} to the two-user / double-state case. As
shown by Jafar \cite{Jafar06}, this extension leads to the
following pentagonal inner bound for the capacity region of
\eqref{eq:GeneralModel}:
\begin{align}
\mathcal{R}(U_1,U_2)\triangleq\Bigg\{(R_1,R_2):\;
&R_1\leq I(U_1,Y|U_2)-I(U_1;S_1)\nonumber\\
&R_2\leq I(U_2,Y|U_1)-I(U_2;S_2)\label{eq:SingleLetterRegion0}\\
&R_1+R_2\leq I(U_1,U_2,Y)-I(U_1;S_1)-I(U_2;S_2)\Bigg\}\nonumber
\end{align}
for some
$P(U_1,U_2,X_1,X_2|S_1,S_2)=P(U_1,X_1|S_1)P(U_2,X_2|S_2)$. In
fact, by a standard time-sharing argument \cite{CoverBook}, the
closure of the convex hull of the set of all rate pairs
$(R_1,R_2)$ satisfying \eqref{eq:SingleLetterRegion},
\begin{align}
\mathcal{R}_{BSL}\triangleq cl\;conv\;
\Bigg\{&(R_1,R_2)\in\mathcal{R}(U_1,U_2):
P(U_1,U_2,X_1,X_2|S_1,S_2)=P(U_1,X_1|S_1)P(U_2,X_2|S_2)\Bigg\},\label{eq:SingleLetterRegion}
\end{align}
is also achievable\footnote{ As in the Gel'fand and Pinsker
solution, for a finite alphabet system it is enough to optimize
over auxiliary variables $U_1$ and $U_2$ whose alphabet size is
bounded in terms of the size of the input and state alphabets. }.
To the best of our knowledge, the set $\mathcal{R}_{BSL}$ is the
best currently known single-letter characterization for the rate
region of the MAC with side information at the transmitters
(\ref{eq:GeneralModel}), and in particular, for the doubly-dirty
MAC \eqref{eq:BinaryModel}\footnote{For the case where the users
have also a common message $W_0$ to be transmitted jointly by both
encoders, \eqref{eq:SingleLetterRegion} can be improved by adding
another auxiliary random variable $U_0$ which plays the role of
the common auxiliary r.v. in Marton's inner bound  for the
non-degraded broadcast channel \cite{Marton79}. In this case, the
joint distribution of $(U_0,U_1,U_2)$ is given by
$P(U_0,U_1,U_2)=P(U_0)P(U_1|U_0)P(U_2|U_0)$, i.e, $U_1$ and $U_2$
are conditionally independent given $U_0$.}. The achievability of
\eqref{eq:SingleLetterRegion} can be proved, as usual, by an i.i.d
random binning scheme \cite{Jafar06}.

A different method to cancel known interference is by ``linear
strategies'', i.e, binning based on the cosets of a linear code
\cite{ErezShamaiZamir05,ZamirShamaiEreznested,PhilosofKhistiErezZamir07}.
In the sequel, we show that the outer bound
\eqref{eq:DoublyInterferenceOuterBound} can indeed be achieved by
a linear coding scheme. Hence, the set of rate pairs $(R_1,R_2)$
satisfying \eqref{eq:DoublyInterferenceOuterBound} is the capacity
region of the binary doubly-dirty MAC.
In contrast, we show that the single-letter region
\eqref{eq:SingleLetterRegion} is \emph{strictly contained} in this
capacity region. Hence, a random binning scheme based on this
extension of the Gel'fand-Pinsker solution \cite{GelfandPinsker80} is not
optimal for this problem.

A similar observation has been made by Korner-Marton
\cite{KornerMarton79} for the ``two help one'' source coding
problem. For a specific binary version known as the ``modulo-two
sum'' problem, they showed that the minimum possible rate sum is
achieved by a linear coding scheme, while the best known
single-letter expression  for this problem is strictly higher. See
the discussion in \cite[Section IV]{KornerMarton79} and in the end
of Section~\ref{sec:SingleLetter}.

Although the ``\textit{single-letter characterization}'' is a
fundamental concept in information theory, it has not been
generally defined \cite[p.35]{GopinathBook}. Csiszar and Korner
\cite[p.259]{CsiszarBook} suggested to define it through the
notion of \textit{computability}, i.e., a problem has a
single-letter solution if there exists an algorithm which can
decide if a point belongs to an $\varepsilon$-neighborhood of the
achievable rate region with polynomial complexity in
$1/\varepsilon$. Since we are not aware of any other computable
solution to our problem, we shall refer to
\eqref{eq:SingleLetterRegion} as the \textit{``best known
single-letter characterization''}.

An extension of these observations to continuous channels would be
of interest. Costa \cite{Costa83} considered the single-user case
of the dirty channel problem $Y=X+S+Z$, where the interference $S$
and the noise $Z$ are assumed to be i.i.d. Gaussian with variances
$Q$ and $N$, respectively, and the input $X$ is subject to a power
constraint $P$. He showed that in this case, the transmitter
side-information capacity \eqref{eq:GelfandPinskerCapacity}
coincides with the zero-interference capacity
$\frac{1}{2}\log_2(1+SNR)$, where $SNR=P/N$. Selecting the
auxiliary random variable $U$ in \eqref{eq:GelfandPinskerCapacity}
such that
\begin{eqnarray}
U=X+\alpha S,\label{eq:CostaRandomBinning}
\end{eqnarray}
where $X$ and $S$ are independent,
% where $X\sim\mathcal{N}(0,P)$ and $S\sim\mathcal{N}(0,Q)$ are
% independent,
and taking $\alpha=\frac{P}{P+N}$, the formula
\eqref{eq:GelfandPinskerCapacity} and its associated random
binning scheme are capacity achieving.
The continuous (Gaussian) version of the doubly-dirty MAC of
\figref{DoublyDirtyMACModel} was considered in
\cite{PhilosofKhistiErezZamir07}. It was shown that by using a
linear structure, i.e., lattice strategies
\cite{ErezShamaiZamir05}, the full capacity region is achieved in
the limit of high SNR and high lattice dimension.
In contrast, it was shown that for $Q \rightarrow \infty$
no positive rate is achievable
by using the natural generalization of Costa's
strategy \eqref{eq:CostaRandomBinning} to the two user case,
while a (scalar) modulo addition version of
\eqref{eq:CostaRandomBinning} looses $\thickapprox 0.254$ bit in
the sum capacity.
We shall further elaborate on this issue in
Section~\ref{sec:GaussianCase}.

Similar observations regarding the advantage of modulo-lattice
modulation with respect to a separation based solution were made
by Nazer and Gastpar \cite{NazerGastpar07}, in the context of
computation over linear Gaussian networks, and also by Krithivasan
and Pradhan \cite{KrithivasanPradhan07} for multi-terminal rate
distortion problems.

The paper is organized as follows. In
Section~\ref{sec:BinaryCapacityRegion} the capacity region for the
binary doubly-dirty MAC \eqref{eq:BinaryModel} is derived, and
linear coding is shown to be optimal.
Section~\ref{sec:SingleLetter} develops a closed form expression
for the best known single-letter characterization
\eqref{eq:SingleLetterRegion} for this channel, and demonstrates
that it is strictly contained in the the true capacity region. In
Section~\ref{sec:GaussianCase} we consider the Gaussian
doubly-dirty MAC, and state a conjecture regarding the capacity
loss of single-letter characterization in this case.

%%%%%%%%%%%%%%%%%%%%%%%%%%%%%%%%%%%%%%%%%%%%%%%%%%%%%%%%%%%%%%%%%%%%%%%%%%%

\section{The Capacity Region of the Binary Doubly-Dirty MAC}\label{sec:BinaryCapacityRegion}
The following theorem characterizes the capacity region of the
binary doubly-dirty MAC of \figref{DoublyDirtyMACModel}.
\begin{theorem}
\label{thm:Capacity}
The capacity region of the binary doubly-dirty MAC
(\ref{eq:BinaryModel}) is the set of all rate pairs $(R_1,R_2)$
satisfying
\begin{align}
\mathcal{C}(q_1,q_2)\triangleq \Bigg\{(R_1,R_2):R_1+R_2\leq
\min\Big\{H_b(q_1),H_b(q_2)\Big\}\Bigg\} .
\label{eq:Capacity}
\end{align}
\end{theorem}

\begin{proof}
\textbf{\textit{The converse part:}} As explained in the
Introduction \eqref{eq:DoublyInterferenceOuterBound}, one way to
derive an upper bound for the rate-sum is through the general
one-dirty-user capacity formula \cite{BaruchShamaiVerdu06}, which
we derive explicitly for the binary case in
Appendix~\ref{App:Appendix1}.
Here we show directly the converse part, which is similar to the
proof of the outer bound for the Gaussian case in
\cite{KhistiPr,PhilosofKhistiErezZamir07}. We assume that user $1$
and user $2$ intend to transmit a common message $W$. An upper
bound on the rate of this message clearly upper bounds the sum
rate $R_1 + R_2$ in the individual messages case. Thus,
\begin{align}
n(R_1+R_2) &\le H(W)\nonumber\\
&= H(W|Y^n) + I(W;Y^n)\nonumber\\
&\le I(W;Y^n) + n\epsilon_n\label{eq:200}\\
&= H(Y^n) - H(Y^n|W) + n\epsilon_n \nonumber\\
&= H(Y^n) - H(Y^n|W, S_1^n, S_2^n) - I(S_1^n, S_2^n; Y^n | W) + n\epsilon_n \nonumber\\
&= H(Y^n) - I(S_1^n, S_2^n; Y^n | W) + n\epsilon_n\label{eq:205}\\
&= H(Y^n) - H(S_1^n, S_2^n|W) + H(S_1^n, S_2^n|W, Y^n) + n\epsilon_n\nonumber\\
&\le -n+ H(S_1^n|W, Y^n) +  H(S_2^n|W, Y^n, S_1^n) + n\epsilon_n \label{eq:210}\\
&\le H(X_1^n \oplus X_2^n \oplus S_1^n|W, Y^n, S_1^n) + n\epsilon_n\label{eq:220}\\
&= H(X_2^n|W, Y^n, S_1^n) + n\epsilon_n\label{eq:225}\\
&\le nH_b(q_2) + n\epsilon_n,\label{eq:230}
\end{align}
where \eqref{eq:200} follows from Fano's inequality where
$\epsilon_n\rightarrow 0$ as the error probability $P_e^{(n)}$
goes to zero for $n\rightarrow\infty$; \eqref{eq:205} follows
since $Y$ is fully known given $W$, $S_1$ and $S_2$;
\eqref{eq:210} follows from the chain rule for entropy, and due to
$H(Y^n) \leq n$ and $H(S_1^n, S_2^n|W)=H(S_1^n)+H(S_2^n)=2n$ since
$W$, $S_1^n$ and $S_2^n$ are mutually independent; \eqref{eq:220}
follows since $H(S_1^n|W, Y^n)\leq n$ and $Y^n=X_1^n \oplus X_2^n
\oplus S_1^n \oplus S_2^n$; \eqref{eq:225} follows since $X_1^n$
is a function of $(W,S_1^n)$, finally \eqref{eq:230} follows since
$H(X_2^n|W, Y^n, S_1^n)\leq H(X_2^n)\leq nH_b(q_2)$.

In the same way we can show that $R_1+R_2\leq
H_b(q_1)+\epsilon_n$. The converse part follows since for
$n\rightarrow\infty$ we have that $\epsilon_n\rightarrow 0$, thus
$P_e^{(n)}\rightarrow 0$.

\textbf{\textit{The direct part}} is based on the scheme for the
point-to-point binary dirty paper channel
\cite{ZamirShamaiEreznested}. We define
$q\triangleq\min\{q_1,q_2\}$. In view of the converse part, it is
sufficient to show achievability of the point
$(R_1,R_2)=(H_b(q),0)$, since the outer bound may be achieved by
time sharing with the symmetric point $(R_1,R_2)=(0,H_b(q))$. The
corner point $(R_1,R_2)=(H_b(q),0)$ corresponds to the ``helper
problem'', i.e., user $2$ tries to help user $1$ to transmit at
its highest rate. The encoders and decoder are described using a
binary linear code $\mathcal{C}(n,k)$ with parity check matrix
$H$. Let $\mathbf{v}\in\mathbb{Z}_2^{n-k}$ be a syndrome of the
code $\mathcal{C}$, where we note that each syndrome represents a
different coset of the linear code $\mathcal{C}$. Let
$f(\mathbf{v})$ denote the ``leader'' of (or the minimum weight
vector in) the coset associated with the syndrome $\mathbf{v}$
\cite[Chap. 6]{Gallager68}, hence
$f:\{0,1\}^{n-k}\rightarrow\{0,1\}^{n}$ . For
$\mathbf{a}\in\mathbb{Z}_2^n$, we define the $n$-dimensional
modulo operation over the code $\mathcal{C}$ as
\begin{align*}
\mathbf{a}\;\mbox{mod}\;\mathbb{C}\triangleq f(H\mathbf{a}),
\end{align*}
which is the leader of the coset to which the vector $\mathbf{a}$
belongs.
\begin{itemize}
\item \textbf{Encoder of user $\mathbf{1}$:} Let the transmitted
message $\mathbf{v}_1\in \mathbb{Z}_2^{n-k}$ be a syndrome in
$\mathcal{C}$, and let $\mathbf{\tilde{x}}_1=f(\mathbf{v}_1)$ be
its coset leader. In particular $\mathbf{v}_1 =
H\mathbf{\tilde{x}}_1$. Transmit the modulo of the code
$\mathcal{C}$ with respect to the difference between
$\tilde{\mathbf{x}}_1$ and $\mathbf{s}_1$, i.e.,
\begin{align*}
\mathbf{x}_1=(\tilde{\mathbf{x}}_1 \oplus
\mathbf{s}_1)\;\mbox{mod}\;\mathbb{C}=f(\mathbf{v}_1\oplus
H\mathbf{s}_1).
\end{align*}

\item \textbf{Encoder of user $\mathbf{2}$:} (functions as a
``helper'' for user $1$). Transmit
\begin{align*}
\mathbf{x}_2=\mathbf{s}_2\;\mbox{mod}\;\mathbb{C}=f(H\mathbf{s}_2).
\end{align*}

\item \textbf{Decoder:}\\
1. Reconstruct $\tilde{\mathbf{x}}_1$ by
$\hat{\tilde{\mathbf{x}}}_1= \mathbf{y}\;\rm{mod}\;\mathbb{C}$.\\
2. Reconstruct the transmitted coset of user $1$ by
$\hat{\mathbf{v}}_1=H\hat{\tilde{\mathbf{x}}}_1$.\\
In fact, the transmitted coset can be reconstructed directly as
$\hat{\mathbf{v}}_1=H\hat{\tilde{\mathbf{x}}}_1=H(\mathbf{y}\;\rm{mod}\;\mathbb{C})=H\mathbf{y}$,
where the last equality follows since
$\mathbf{y}\;\rm{mod}\;\mathbb{C}$ and $\mathbf{y}$ are in the
same coset.

\end{itemize}
It follows that the decoder correctly decodes the message coset
$\mathbf{v}_1$, since
\begin{align*}
\hat{\mathbf{v}}_1 &= H\cdot\Big(\mathbf{y}\;\rm{mod}\;\mathbb{C}\Big)\\
&= H\cdot\Big([\tilde{\mathbf{x}}_1 \oplus
\mathbf{s}_1 \oplus \mathbf{s}_2 \oplus \mathbf{s}_1 \oplus \mathbf{s}_2]\;\mbox{mod}\;\mathbb{C}\Big)\\
&= H\tilde{\mathbf{x}}_1\\
&=\mathbf{v}_1,
\end{align*}
where the third equality follows since $\tilde{\mathbf{x}}_1$ and
$\tilde{\mathbf{x}}_1\;\rm{mod}\;\mathbb{C}$ are in the same
coset. It is left to relate the coding rate
$R_1=\frac{1}{n}\log\Big(\Big|\{0,1\}^{n-k}\Big|\Big)=1-k/n$ to
the input constraint $q$. Form \cite{CoveringBook}, there exists a
binary linear code with covering radius $\rho$ that satisfies
$\frac{k}{n}\leq 1-H_b(\rho/n)+\epsilon$ where
$\epsilon\rightarrow 0$ as $n\rightarrow\infty$. The achievability
of the point $(H_b(q),0)$ follows by using $q=\rho/n$, thus
$R_1=1-k/n\geq H_b(q)-\epsilon$, while
$w_H(\mathbf{x}_1)=w_H(f(\mathbf{v}_1\oplus H\mathbf{s}_1))\leq
\rho$ and $w_H(\mathbf{x}_2)=w_H(f(H\mathbf{s}_2))\leq \rho$,
hence
\begin{align*}
&\frac{1}{n}Ew_H\{\mathbf{x}_1\}=\frac{1}{n}Ew_H\{f(\mathbf{v}_1\oplus H\mathbf{s}_1)\}\leq q\\
&\frac{1}{n}Ew_H\{\mathbf{x}_2\}=\frac{1}{n}Ew_H\{f(H\mathbf{s}_2)\}\leq
q.
\end{align*}
This completes the proof of the direct part of the theorem.
\end{proof}

\vspace{5mm}

As stated above, the achievability for the capacity region follows
by time sharing the corner points $(H_b(q),0)$ and
$(0,H_b(q))$ where $q=\min\{q_1,q_2\}$.
It is also interesting to
see how to achieve the rate sum $H_b(q)$ for an arbitrary rate
pair $(R_1,R_2)$ without time sharing.
For that, let the message of user $1$ be
$\mathbf{m}_1\in \mathbb{Z}_2^{l_1}$ and the message of user $2$
be $\mathbf{m}_2\in \mathbb{Z}_2^{l_2}$ where $l_1+l_2=n-k$. We
define the following syndromes in $\mathcal{C}$
\begin{align*}
\mathbf{v}_1&\triangleq[\mathbf{m}_1\;\underbrace{0\;0\;\ldots\;0}_{l_2}]\in \mathbb{Z}_2^{n-k}\\
\mathbf{v}_2&\triangleq[\underbrace{0\;0\;\ldots\;0
}_{l_1}\;\mathbf{m}_2]\in \mathbb{Z}_2^{n-k}\\
\mathbf{v}&\triangleq\mathbf{v}_1 \oplus \mathbf{v}_2.
\end{align*}
Clearly, given the syndrome $\mathbf{v}$ the syndromes
$\mathbf{v}_1$ and $\mathbf{v}_2$ are fully known and the messages
$\mathbf{m}_1$ and $\mathbf{m}_2$ as well. Let
$\mathbf{\tilde{x}}_i=f(\mathbf{v}_i)$ be the coset leader of
$\mathbf{v}_i$ for $i=1,2$. In this case the transmission scheme
is as follow:
\begin{itemize}
\item Encoder of user $1$: transmit $
\mathbf{x}_1=(\tilde{\mathbf{x}}_1 \oplus
\mathbf{s}_1)\;\mbox{mod}\;\mathbb{C}=f(\mathbf{v}_1\oplus
H\mathbf{s}_1) $.

\item Encoder of user $2$: transmit $
\mathbf{x}_2=(\tilde{\mathbf{x}}_2 \oplus
\mathbf{s}_2)\;\mbox{mod}\;\mathbb{C}=f(\mathbf{v}_2\oplus
H\mathbf{s}_2) $.

\item Decoder: reconstruct $\hat{\mathbf{v}}
=H\cdot\Big(\mathbf{y}\;\rm{mod}\;\mathbb{C}\Big)$.
\end{itemize}
Therefore, we have that
\begin{align*}
\hat{\mathbf{v}} &=
H\cdot\Big(\mathbf{y}\;\rm{mod}\;\mathbb{C}\Big) \\
&= H\cdot\Big(\tilde{\mathbf{x}}_1 \oplus
\tilde{\mathbf{x}}_2\Big) =\mathbf{v}_1 \oplus \mathbf{v}_2
=\mathbf{v}.
\end{align*}
The sum capacity is achieved since
$R_1+R_2=\frac{l_1+l_2}{n}=\frac{n-k}{n}\geq H_b(q)-\epsilon$
where $\epsilon\rightarrow 0$ as $n\rightarrow\infty$ which
satisfies the input constraints.

%%%%%%%%%%%%%%%%%%%%%%%%%%%%%%%%%%%%%%%%%%%%%%%%%%%%%%%%%%%%%%%%%%%%%%%%%%%

\section{A Single-Letter Characterization for the Capacity Region}\label{sec:SingleLetter}

In this section we characterize the best known single-letter
region \eqref{eq:SingleLetterRegion} for the binary doubly-dirty
MAC \eqref{eq:BinaryModel}, and show that it is strictly contained
in the capacity region \eqref{eq:Capacity}. For simplicity, we
shall assume identical input constraints, i.e., $q_1=q_2=q$.
\begin{definition}
For a given $q$, the best known single-letter rate region for the
binary doubly-dirty MAC \eqref{eq:BinaryModel}, denoted by
$\mathcal{R}_{BSL}(q)$, is the set of all rate pairs $(R_1,R_2)$
satisfying \eqref{eq:SingleLetterRegion} with the additional
constraints that $E\mathbf{X}_1,E\mathbf{X}_2 \leq q$.
\end{definition}

In the following theorem we give a closed form expression for
$\mathcal{R}_{BSL}(q)$.
\begin{theorem}\label{thm:BSLRate}
The best known single-letter rate region for the binary
doubly-dirty MAC \eqref{eq:BinaryModel} is a triangular region
given by
\begin{align}
&\mathcal{R}_{BSL}(q) = \Bigg\{(R_1,R_2):R_1+R_2\leq u.c.e
\Big[2H_b(q)-1\Big]^+\Bigg\},\label{eq:SumRateBound2}
\end{align}
where $u.c.e$ is the upper convex envelope with respect to $q$,
and $[x]^+\triangleq\max\{0,x\}$.
\end{theorem}

\figref{BinaryDMAC} shows the sum capacity of the binary
doubly-dirty MAC \eqref{eq:Capacity} versus the best known
single-letter rate sum \eqref{eq:SumRateBound2} for equal input
constraints. The latter is strictly contained in the capacity
region which is achieved by a linear code. The quantity
$[2H_b(q)-1]^+$ is not a convex - $\cap$ function with respect to
$q$. The upper convex envelope of $[2H_b(q)-1]^+$ is achieved by
time-sharing between the points $q=0$ and $q=q^*\triangleq
1-1/\sqrt{2}$, therefore it is given by
\begin{align}\label{eq:SumCapacityBSL}
R_1+R_2\leq\Bigg\{
\begin{array}{cc}
  2H_b(q)-1, & q^*\leq q\leq 1/2 \\
  C^*q, & 0\leq q\leq q^* \\
\end{array},
\end{align}
where $C^*\triangleq\frac{2H_b(q^*)-1}{q^*}$.

\begin{proof}
\textbf{\textit{The direct part}} is shown by choosing in
\eqref{eq:SingleLetterRegion0} $U_1=S_1 \oplus X_1$ and $U_2=S_2
\oplus X_2$, where $X_1,X_2\sim \mbox{Bernoulli}(q)$ and
$X_1,X_2,S_1,S_2$ are independent. From
\eqref{eq:SingleLetterRegion0} the achievable rate sum is given by
\begin{align}
&R_1+R_2 = I(U_1,U_2;Y)-I(U_1,U_2;S_1,S_2)\nonumber\\
&=H(U_1|S_1)+H(U_2|S_2)-H(U_1,U_2|U_1 \oplus U_2)\label{eq:600}\\
&=H(U_1|S_1)+H(U_2|S_2)-H(U_1|U_1 \oplus U_2)-H(U_2|U_1 \oplus U_2,U_1)\label{eq:605}\\
&=H(X_1)+H(X_2)-H(U_1|U_1 \oplus U_2)\label{eq:610}\\
&=2H_b(q)-1,\label{eq:620}
\end{align}
where \eqref{eq:600} follows since $Y=U_1 \oplus U_2$;
\eqref{eq:605} follows from the chain rule for entropy;
\eqref{eq:610} follows since $U_2$ is fully known given $U_1
\oplus U_2,U_1$ thus $H(U_2|U_1 \oplus U_2,U_1)=0$; \eqref{eq:620}
follows since $H(X_i)\leq H_b(q)$ and since $U_1,U_2$ are
independent with $P(U_i=1)=1/2$ thus $H(U_1|U_1 \oplus
U_2)=H(U_1)=1$.

\textbf{\textit{The converse part}} of the proof is given in
Appendix \ref{App:Appendix2}.
\end{proof}

\begin{figure}[h]
\begin{center}
\epsfig{file=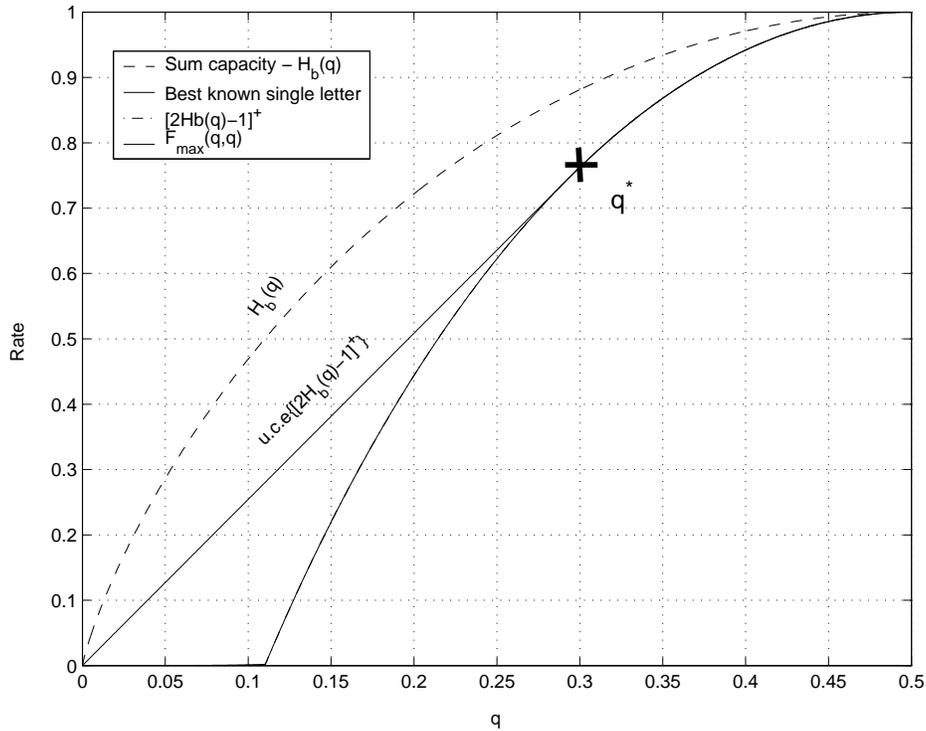,scale=0.7} \caption{The rate sum of
binary doubly-dirty MAC vs. best known single-letter rate sum with
input constraints $EX_1,EX_2\leq q$.}\label{fig:BinaryDMAC}
\end{center}
\end{figure}

We see that the binary doubly-dirty MAC is a memoryless channel
coding problem, where the capacity region is achievable by a
linear code, while the best known single-letter rate region is
strictly contained in the capacity region. This may be explained
by the fact that each user has only \emph{partial} side
information, and distributed random binning is unable to capture
the linear structure of the channel.

In order to understand the limitation of random binning versus a
linear code, we consider these two schemes for high enough $q$,
that is $2H_b(q)-1 \geq 0$. The random binning scheme uses
$U_i=X_i \oplus S_i$ where $X_i\sim\mbox{Bernoulli}(q)$ and
$S_i\sim\mbox{Bernoulli}(1/2)$ are independent, therefore $Y=U_1
\oplus U_2$ where $U_i\sim\mbox{Bernoulli}(1/2)$ for $i=1,2$. Each
transmitter maps the message (bin) $W_i$ into a codeword
$\mathbf{u}_i$ which is with high probability at a Hamming
distance of $nq$ from $\mathbf{s}_i$. Therefore, given the vectors
$(\mathbf{s}_1^n,\mathbf{s}_2^n)$, the available \emph{input
space} is approximately
$2^{nH(U_1,U_2|S_1,S_2)}=2^{nH(X_1,X_2)}=2^{2nH_b(q)}$. Given the
received vector $\mathbf{y}$, the \emph{residual ambiguity} is
given by $2^{nH(U_1,U_2|Y)}=2^{n[H(U_1|Y)+H(U_2|Y,U_1)]}=2^{n}$,
since $H(U_1|Y)=1$ and $H(U_2|Y,U_1)=0$. As a result, the
achievable rate sum is given by
\begin{align*}
R_1+R_2 =
\frac{1}{n}\log_2\Big(\frac{|\mbox{input
space}|}{|\mbox{residual ambiguity space}|}\Big)\approx 2H_b(q)-1.
\end{align*}
The linear coding scheme shown in \thrmref{Capacity} has the same
input space size as the random binning scheme, i.e.,
$2^{2nH_b(q)}$, since each user has $2^{nH_b(q)}$ cosets. However,
given the received vector $\mathbf{y}$ there are $2^{nH_b(q)}$
possible pairs of cosets, i.e., the residual ambiguity is only
$2^{nH_b(q)}$. Therefore, the linear code achieves rate sum of
$R_1+R_2 \approx 2H_b(q)-H_b(q)=H_b(q)$. The advantage of the
linear coding scheme results from the ``ordered structure'' of the
linear code, which decreases the residual ambiguity from $1$ bit
in random coding to $H_b(q)$.

The following example illustrates the above arguments for the case
that user $2$ is a ``helper'' for user $1$, i.e, $R_2=0$, and user
$1$ transmits at his highest rate for each technique (random
binning or linear coding). Table \ref{tab:RandomBinningLinearCode}
summarizes the rates and codebooks sizes for each user for
$q=0.3$, that is $H_b(q)\thickapprox 0.88$ bit.

\begin{table}[h]
\begin{center}
\begin{tabular}{|c|c|c|}
  \hline
  % after \\: \hline or \cline{col1-col2} \cline{col3-col4} ...
  & Random binning & Linear code \\
  \hline
  Rate sum & $2H_b(q)-1=0.76$ bit & $H_b(q)=0.88$ bit\\
  \hline
  Codewords per bin/coset &  $2^{nI(U_i;S_i)}=2^{n[1-H_b(q)]}=2^{0.12n}$ & $2^{n[1-H_b(q)]}=2^{0.12n}$\\
  \hline
  Helper (user $2$) - codebook size & $2^{nI(U_2;S_2)}=2^{n[1-H_b(q)]}=2^{0.12n}$ & $2^{n[1-H_b(q)]}=2^{0.12n}$\\
  \hline
  User $1$ - codebook size & $2^{0.76n}2^{0.12n}=2^{0.88n}$ & $2^{0.12n}2^{0.88n}=2^{n}$\\
  \hline
  Number of possible codeword pairs & $2^{0.88n}2^{0.12n}=2^{n}$ & $2^{n}2^{0.12n}=2^{1.12n}$\\
  \hline
\end{tabular}
\end{center}
\caption{Random binning and linear coding schemes codebooks sizes
for the helper problem with
$q=0.3$.}\label{tab:RandomBinningLinearCode}
\end{table}

\begin{figure}[t]
%\vspace{+0.5cm}
  \begin{center}
    \input{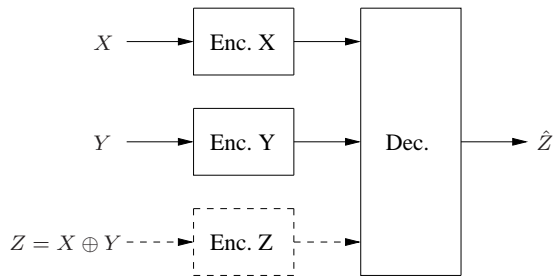}
%    \vspace{-0.2cm}
    \caption{The Korner-Marton configuration.}\label{fig:KMModel}
  \end{center}
%\vspace{-0.2cm}
\end{figure}

Korner and Marton \cite{KornerMarton79} observed a similar
behavior for the ``two help one'' source coding problem shown in
\figref{KMModel}. In this problem, there are three binary sources
$X,Y,Z$, where $Z=X \oplus Y$, and the joint distribution of $X$
and $Y$ is symmetric with $P(X\neq Y)=\theta$. The goal is to
encode the sources $X$ and $Y$ separately such that $Z$ can be
reconstructed losslessly. Korner and Marton showed that the rate
sum required is at least
\begin{align}
R_x+R_y \geq 2H(Z),\label{eq:KornerMartonRateSum}
\end{align}
and furthermore, this rate sum can be achieved by a linear code:
each encoder transmits the syndrome of the observed source
relative to a good linear binary code for a BSC with crossover
probability $\theta$.

In contrast, the ``one help one'' problem
\cite{AhlswedeKorner75,Wyner75} has a closed single-letter
expression for the rate region, which corresponds to a random
binning coding scheme. Korner and Marton \cite{KornerMarton79}
generalize the expression of \cite{AhlswedeKorner75,Wyner75} to
the ``two help one'' problem, and show that the minimal rate sum
required using this expression is given by
\begin{align}
R_x+R_y \geq H(X,Y).\label{eq:SlepianWolfSumRate}
\end{align}
The region \eqref{eq:SlepianWolfSumRate} corresponds to
Slepian-Wolf encoding of $X$ and $Y$, and it can also be derived
from the Burger-Tung achievable region \cite{Berger77} for
distributed coding for $X$ and $Y$ with one reconstruction
$\hat{Z}$ under the distortion measure $d(X,Y,\hat{Z})\triangleq X
\oplus Y \oplus \hat{Z}$. Clearly, the region
\eqref{eq:SingleLetterRegion0} is strictly contained in the
Korner-Marton region $R_x+R_y \geq 2H(Z)$
\eqref{eq:KornerMartonRateSum} (since $H(X,Y)=1+H(Z)> 2H(Z)$ for
$Z\sim \mbox{Bernoulli}(\theta)$, where $\theta \neq
\frac{1}{2}$). For further background on related source coding
problems, see \cite{KrithivasanPradhan07}.

%%%%%%%%%%%%%%%%%%%%%%%%%%%%%%%%%%%%%%%%%%%%%%%%%%%%%%%%%%%%%%%%%%%%%%%%%%%

%%%%%%%%%%%%%%%%%%%%%%%%%%%%%%%%%%%%%%%%%%%%%%%%%%%%%%%%%%%%%%%%%%%%%%%%%%%

\section{The Gaussian Doubly-Dirty MAC}\label{sec:GaussianCase}
In this section we introduce our conjecture regarding the rate
loss of the best known single-letter characterization for the
capacity region of the two-user Gaussian doubly-dirty MAC at high
SNR. The Gaussian doubly-dirty MAC
\cite{PhilosofKhistiErezZamir07} is given by
\begin{eqnarray}
Y=X_1+X_2+S_1+S_2+Z\label{eq:GaussianDirtyMAC},
\end{eqnarray}
where $Z\sim\mathcal{N}(0,N)$ is independent of $X_1,X_2,S_1,S_2$,
and where user $1$ and user $2$ must satisfy the power
constraints, $\frac{1}{n}\sum_{i=1}^{n}X_{1_i}^2\leq P_1$ and
$\frac{1}{n}\sum_{i=1}^{n}X_{2_i}^2\leq P_2$ see
\figref{DoublyDirtyMACModel}.
The interference signals $S_1$ and $S_2$
are known non-causally to the transmitters of user $1$ and user
$2$, respectively. We shall assume that $S_1$ and $S_2$ are
independent Gaussian with variances going to infinity, i.e.,
$S_i\sim\mathcal{N}(0,Q_i)$ where $Q_i\rightarrow\infty$ for
$i=1,2$. The signal to noise ratios for the two users are
$SNR_1=\frac{P_1}{N}$ and $SNR_2=\frac{P_2}{N}$.

The capacity region at high SNR, i.e., $SNR_1,SNR_2 \gg 1$, is
given by \cite{PhilosofKhistiErezZamir07},
\begin{align}
R_1+R_2\leq
\frac{1}{2}\log_2\Bigg(\frac{\min\{P_1,P_2\}}{N}\Bigg),\label{eq:GaussianDMACCapacity}
\end{align}
and it is achievable by a modulo lattice coding scheme of
dimension going to infinity. In contrast, it was shown in
\cite{PhilosofKhistiErezZamir07} that at high SNR and strong
independent Gaussian interferences, the natural generalization of
Costa's strategy (\ref{eq:CostaRandomBinning}) for the two users
case, i.e., with auxiliary random variables $U_1=X_1+S_1$ and
$U_2=X_2+S_2$, is not able to achieve \emph{any positive rate}. A
better choice for $U_1$ and $U_2$ suggested in
\cite{PhilosofKhistiErezZamir07} is a modulo version of Costa's
strategy (\ref{eq:CostaRandomBinning}),
\begin{align}
U_i^*=[X_i+S_i]\;\mbox{mod}\:\Delta_i ,
\label{eq:ScalarModulo}
\end{align}
where $\Delta_i=\sqrt{12P_i}$, and where
$X_i\sim\mbox{Unif}\left([-\frac{\Delta_i}{2},\frac{\Delta_i}{2})\right)$
is independent of $S_i$, for $i=1,2$. In this case the rate loss
with respect to \eqref{eq:GaussianDMACCapacity} is
$\frac{1}{2}\log_2\Big(\frac{\pi e}{6}\Big)\thickapprox
0.254\;\rm{bit}$.

The best known single-letter capacity region for the Gaussian
doubly-dirty MAC \eqref{eq:GaussianDirtyMAC} is defined as the set
of all rate pairs $(R_1,R_2)$ satisfying
\eqref{eq:SingleLetterRegion}, where $X_1$ and $X_2$ are
restricted to the power constraints $EX_1^2\leq P_1$ and
$EX_2^2\leq P_2$.
We believe that for high SNR and strong interference,
the modulo-$\Delta$ strategy
(\ref{eq:ScalarModulo})
is an optimum choice
for $(X_1,X_2,U_1,U_2)$ in \eqref{eq:SingleLetterRegion} for the
Gaussian doubly-dirty MAC.
This implies the following conjecture
about the rate loss  of the best known single-letter characterization.

\begin{conjecture}
For the Gaussian doubly-dirty MAC, at high SNR and strong
interference, the best known single-letter expression
$R_{BSL}^{sum}$ \eqref{eq:SingleLetterRegion} looses
\begin{align}
C^{sum}-R_{BSL}^{sum}=\frac{1}{2}\log_2\Big(\frac{\pi
e}{6}\Big)\approx 0.254\;\mbox{bit},\label{eq:ShapingLoss}
\end{align}
with respect to the sum capacity $C^{sum}$
\eqref{eq:GaussianDMACCapacity}.
\end{conjecture}
Note that the right hand side of \eqref{eq:ShapingLoss} is the
well known ``shaping loss'' \cite{Forney89a} (equivalent to a
$~1.53dB$ power loss).

A heuristic approach to attack the proof of this conjecture is to
follow the steps of the proof of the converse part in the binary
case (Theorem \ref{thm:BSLRate}). First, in
\lemref{GaussianBSLRate} we derive a simplified single-letter
formula, $\overline{G}_{max}(P_1,P_2)$, which is analogous to
\lemref{EquivalentBSL} in the binary case.
The next step would be to optimize this expression.
However, an optimal choice for the auxiliary random variables
$V_1,V_1',V_2,V_2'$
(provided in the binary case by
\lemref{EquivalentBSL1} and \lemref{OuterBound})
is unfortunately still missing for the Gaussian case.
The expression in
\lemref{GaussianBSLRate} is close in spirit to the point-to-point
\emph{dirty tape} capacity for high SNR and strong interference
\cite{ErezShamaiZamir05}. In \cite{ErezShamaiZamir05} it is shown
that optimizing the capacity is equivalent to minimum
\emph{entropy-constrained scalar quantization} in high resolution,
which is achieved by a lattice quantizer. Clearly, if we could
show a similar lemma for the two variable pairs in the
maximization of \lemref{GaussianBSLRate}, i.e., that it is
achieved by a pair of lattice quantizers, then the conjecture
would be an immediate consequence.

It should be noted that the above discussion is valid only for
strong interferences $S_1$ and $S_2$. For interference with finite
power, it seems that cancelling the interference part of the time
and staying silence the rest of the time (like in the time-sharing
region $ 0\leq q \leq q^*$ in the binary case) may achieve better
rates.

%%%%%%%%%%%%%%%%%%%%%%%%%%%%%%%%%%%%%%%%%%%%%%%%%%%%%%%%%%%%%%%%%%%%%%%%%%%
\section{Summary}\label{sec:Summary}

A memoryless information theoretic problem is considered open as
long as we are missing a general single-letter characterization
for its information performance. This goes hand in hand with the
optimality of the random coding approach for those problems which
are currently solved. We examined this traditional view for the
memoryless doubly-dirty MAC.

In the binary case, we showed that the best known single letter
characterization is strictly contained in the region achievable by
linear coding, and that the latter is in fact the full capacity
region of the problem.
In the Gaussian case, we conjectured that the best known
single-letter characterization suffers an inherent rate loss
(equal to the well known ``shaping loss'' $0.5 \log(\pi e /6)$),
and we provide a partial proof. This is in contrast to the
asymptotic optimality (dimension $\rightarrow \infty$) of lattice
strategies, as recently shown in \cite{PhilosofKhistiErezZamir07}.

The underlying reason for these performance gaps   %% of single letter characterizations
is that
% These observations indicate that the traditional random binning approach
{\em random} binning is in general not optimal when side information is distributed
among more than one terminal in the network.
In the specific case of the doubly-dirty MAC
(like in Korner-Marton's modulo-two sum problem
\cite{KornerMarton79} and similar settings \cite{NazerGastpar07,KrithivasanPradhan07}),
the linear structure of the network allows to show that {\em linear binning}
is not only better, but it is capacity achieving.

%%%%%%%%%%%%%%%%%%%%%%%%%%%%%%%%%%%%%55%%%%%%%%%%%%%%%%%%%%%%%%%%%%%%%%%%%%%%%%%%%%%%
\appendices
%%%%%%%%%%%%%%%%%%%%%%%%%%%%%%%%%%%%%55%%%%%%%%%%%%%%%%%%%%%%%%%%%%%%%%%%%%%%%%%%%%%%
\section{A Closed Form Expression for the Capacity of the Binary MAC with One Dirty User}\label{App:Appendix1}
We consider the binary dirty MAC \eqref{eq:BinaryModel} with
$S_2=0$,
\begin{align}\label{eq:OneDirtyUserModel}
Y=X_1 \oplus X_2 \oplus S_1,
\end{align}
where $S_1\sim\text{Bernoulli(1/2)}$ is known non-causally at the
encoder of user $1$ with the input constraints
$\frac{1}{n}W_H(\mathbf{x}_i)\leq q_i$ for $i=1,2$.
We show that
%
% \begin{lemma}\label{lem:CapacitySingleInterference}
%
the common message ($W_1=W_2=W$) capacity %% $C_{com}$
of this channel is given by
% for the binary dirty MAC with one dirty user \eqref{eq:OneDirtyUserModel} is given by
\begin{align}
\label{eq:CapacitySingleInterference}
C_{com}=H_b(q_1).
\end{align}

% \end{lemma}
% \begin{proof}

To prove (\ref{eq:CapacitySingleInterference}),
consider the general expression for
the common message capacity of the MAC with one informed
user \cite{BaruchShamaiVerdu06}, given by
\begin{align}\label{eq:OneDirtyUserCapapcity}
C_{com}=
\max_{U_1,X_1,X_2}\{I(U_1,X_2;Y)-I(U_1,X_2;S_1)\},
\end{align}
where the maximization is over al the joint distributions
\begin{align*}
P(S_1,X_1,X_2,U_1,Y)=P(S_1)P(X_2)P(U_1|X_2,S_1)P(X_1|S_1,U_1)P(Y|X_1,X_2,S_1).
\end{align*}
\textbf{\textit{The converse part}} of (\ref{eq:CapacitySingleInterference})
follows since for any
$U_1,X_1,X_2$, the common message rate $R_{com}$ can be upper
bounded by
\begin{align}
R_{com} &= I(U_1,X_2;Y)-I(U_1,X_2;S_1)\nonumber\\
&=H(S_1|U_1,X_2)-H(Y|U_1,X_2)+H(Y)-H(S_1)\nonumber\\
&\leq H(S|U_1,X_2)-H(Y|U_1,X_2)\label{eq:1220}\\
&= H(S_1|U_1,X_2)-H(X_1 \oplus S_1|U_1,X_2)\label{eq:1230}\\
&= H(S_1|T)-H(X_1 \oplus S_1|T)\label{eq:1240}\\
&= E_T\Big\{H(S_1|T=t)-H(X_1 \oplus
S_1|T=t)\Big\}\label{eq:1250}\\
&= E_T\Big\{H_b(\alpha_{t})-H_b(\beta_{t})\Big\}\label{eq:1260},
\end{align}
where \eqref{eq:1220} follows since $H(Y)\leq 1$ and $H(S_1)=1$;
\eqref{eq:1230} follows since $Y=X_1 \oplus X_2 \oplus S_1$;
\eqref{eq:1240} follows the definition $T\triangleq(U_1,X_2)$;
\eqref{eq:1250} follows from the definition of the conditional
entropy; \eqref{eq:1260} follows from the following definitions
$\alpha_{t}\triangleq P(S_1=1|T=t)$ and $\beta_{t}\triangleq P(S_1
\oplus X_1 =1|T=t)$ for any $t\in T$. We also define
$q_{1|t}\triangleq P(X_1=1|T=t)=E\{X_1|T=t\}$, therefore the input
constraint of user $1$ can be written as
\begin{align}\label{eq:1310}
EX_1=E_TE\{X_1|T=t\}=E_T\{q_{1|t}\}\leq q_1.
\end{align}
Without loss of generality, we can only consider $\alpha_t,
\beta_{t}, q_{1|t}\in[0,1/2]$ in \eqref{eq:1260} for any $t\in T$.
Thus,
\begin{align}
R_{com}&\leq E_T\Big\{H_b(\alpha_{t})-H_b\Big([\alpha_t-q_{1|t}]^+\Big)\Big\}\label{eq:1270}\\
&\leq E_T\Big\{H_b(q_{1|t})\Big\}\label{eq:1280}\\
&\leq H_b\Big(E_T\{q_{1|t}\}\Big)\label{eq:1290}\\
&\leq H_b(q_1),\label{eq:1300}
\end{align}
where \eqref{eq:1270} follows from \eqref{eq:1260} and since
$H_b(\beta_t)\geq H_b\Big([\alpha_t-q_{1|t}]^+\Big)$, where
$[x]^+=max\{x,0\}$; \eqref{eq:1280} follows since
$H_b(\alpha_{t})-H_b\Big([\alpha_t-q_{1|t}]^+\Big)$ is increasing
in $\alpha_{t}$ for $\alpha_{t}\leq q_{1|t}\leq 1/2$ and
decreasing in $\alpha_{t}$ for $q_{1|t} < \alpha_{t} \leq 1/2$,
thus the maximum is for $\alpha_{t}=q_{1|t}$; \eqref{eq:1290}
follows from Jensen's inequality since $H_b(\cdot)$ is
convex-$\cap$; \eqref{eq:1300} follows from the input constraint
for user $1$ \eqref{eq:1310}. The converse part follows since the
outer bound is valid for any $U_1$ and $X_1,X_2$ that satisfy the
input constraints.

\textbf{\textit{The direct part}} is shown by using $U_1=X_1
\oplus S_1$ where $X_1$ and $S_1$ are independent with
$X_1\sim\text{Bernoulli}(q_1)$, thus
$U_1\sim\text{Bernoulli}(1/2)$. Furthermore,
$X_2\sim\text{Bernoulli}(q_2)$ which is independent of
$X_1,U_1,S_1$. In this case $Y=U_1 \oplus X_2$, hence
$Y\sim\text{Bernoulli}(1/2)$. Using this choice for $U_1,X_1,X_2$,
the achievable common message rate is given by
\begin{align}
R_{com} &=I(U_1,X_2;Y)-I(U_1,X_2;S_1)\nonumber\\
&=H(S_1|U_1,X_2)-H(Y|U_1,X_2)+H(Y)-H(S_1)\nonumber\\
&=H(X_1)\label{eq:1320}\\
&=H_b(q_1),\nonumber
\end{align}
where \eqref{eq:1320} follows since
$H(S_1|U_1,X_2)=H(S_1|U_1)=H(X_1)$, $H(Y|U_1,X_2)=0$, $H(Y)=1$ and
$H(S_1)=1$.
% \end{proof}

%%%%%%%%%%%%%%%%%%%%%%%%%%%%%%%%%%%%%55%%%%%%%%%%%%%%%%%%%%%%%%%%%%%%%%%%%%%%%%%%%%%%
\section{Proof of the Converse Part of Theorem \ref{thm:BSLRate}}\label{App:Appendix2}
The proof of the converse part follows from
\lemref{EquivalentBSL}, \lemref{EquivalentBSL1} and
\lemref{OuterBound}, whereas \lemref{gr} and \lemref{f2} are
technical results which assist in the derivation of
\lemref{OuterBound}.

Let us define the following functions:
\begin{align}
&F(P_{V_1,V'_1},P_{V_2,V'_2})\triangleq \Big[H(V_1)+H(V_2)-H(V'_1
\oplus V'_2)-1\Big]^+,\label{eq:FDefinition}
\end{align}
where  $[x]^+ = \max(0,x)$; its $(q_1,q_2)$-constrained
maximization with respect to $V_1,V'_1,V_2,V'_2\in\mathbb{Z}_2$
where $(V_1,V'_1)$ and $(V_2,V'_2)$ are independent, i.e.,
\begin{align}
F_{max}(q_1,q_2)\triangleq &\max_{V_1,V'_1,V_2,V'_2}F(P_{V_1,V'_1},P_{V_2,V'_2})\label{eq:FmaxDefinition}\\
&\mbox{s.t}\;\;P(V_i \neq V'_i)\leq
q_i,\;\mbox{for}\;\;i=1,2;\nonumber
\end{align}
and the upper convex envelope of $F_{max}(q_1,q_2)$ with respect
to $q_1,q_2$
\begin{align}
&\overline{F}_{max}(q_1,q_2)\triangleq
u.c.e\Big\{F_{max}(q_1,q_2)\Big\}.\label{eq:FuceDefinition}
\end{align}
In the following lemma we give an outer bound for the
single-letter region \eqref{eq:SingleLetterRegion} of the binary
doubly-dirty MAC in the spirit of \cite[Lemma 3]{CohenZamir07} and
\cite[Proposition 1]{ErezShamaiZamir05}.

\begin{lemma}\label{lem:EquivalentBSL}
The best known single-letter rate sum
\eqref{eq:SingleLetterRegion} of the binary doubly-dirty MAC
\eqref{eq:BinaryModel} with input constraint $q_1$ and $q_2$ is
upper bounded by
\begin{align}
R_1+R_2\leq \overline{F}_{max}(q_1,q_2).\label{eq:SumRateBound}
\end{align}
\end{lemma}
\begin{proof}
An outer bound on the best known single-letter region
\eqref{eq:SingleLetterRegion} is given by
\begin{align}
R_{BSL}^{sum}(U_1,U_2)&\triangleq\Big[I(U_1,U_2;Y)-I(U_1,U_2;S_1,S_2)\Big]^+\label{eq:400}\\
&=\Big[H(S_1|U_1)+H(S_2|U_2)-H(Y|U_1,U_2)+H(Y)-H(S_1)-H(S_2)\Big]^+\label{eq:410}\\
&\leq \Big[H(S_1|U_1)+H(S_2|U_2)-H(Y|U_1,U_2)-1\Big]^+\label{eq:420}\\
&=\Bigg[E_{U_1,U_2}\Big\{H(S_1|U_1=u_1)+H(S_2|U_2=u_2)-H(Y|U_1=u_1,U_2=u_2)-1\Big\}\Bigg]^+\label{eq:430}\\
&\leq E_{U_1,U_2}\Bigg\{\Big[H(S_1|U_1=u_1)+H(S_2|U_2=u_2)-H(Y|U_1=u_1,U_2=u_2)-1\Big]^+\Bigg\}\label{eq:435}\\
&\leq E_{U_1,U_2}\Bigg\{F\Big(P_{S_1,S_1 \oplus
X_1|U_1=u_1},P_{S_2,S_2 \oplus
X_2|U_2=u_2}\Big)\Bigg\}\label{eq:440}\\
&\leq E_{U_1,U_2}\Big\{\overline{F}_{max}\big(q_{1|u_1},q_{2|u_2}\big)\Big\}\label{eq:450}\\
&\leq \overline{F}_{max}\big(E_{U_1}q_{1|u_1},E_{U_2}q_{2|u_2}\big)\label{eq:460}\\
&\leq \overline{F}_{max}\big(q_1,q_2\big),\label{eq:470}
\end{align}
where \eqref{eq:420} follows since $H(S_1)=H(S_2)=1$ and $H(Y)\leq
1$; \eqref{eq:430} follows from the definition of the conditional
entropy; \eqref{eq:435} follows since $[Ex]^+\leq E\{x^+\}$;
\eqref{eq:440} follows from the definition of the function
$F(P_{V_1,V'_1},P_{V_2,V'_2})$ \eqref{eq:FDefinition}, likewise
\eqref{eq:450} follows from the definition of the function
$\overline{F}_{max}(q_1,q_2)$ \eqref{eq:FuceDefinition}, and from
the definition
\begin{align*}
q_{i|u_i}\triangleq P(S_i \neq X_i \oplus
S_i|U_i=u_i)=P(X_i=1|U_i=u_i),\;for\;i=1,2;
\end{align*}
\eqref{eq:460} follows from Jensen's inequality since
$\overline{F}_{max}(q_1,q_2)$ is a concave function;
\eqref{eq:470} follows from the input constraints where
\begin{align}
EX_i&=E_{U_i}P(X_i = 1 | U_i=u_i)\nonumber\\
&=\sum_{u_i\in U_i}P(u_i)P(X_i=1|U_i=u_i)\nonumber\\
&=\sum_{u_i\in U_i}P(u_i)q_{i|u_i}\leq
q_i,\;\mbox{for}\;\;i=1,2.\label{eq:InputConst}
\end{align}
The lemma now follows since the upper bound \eqref{eq:470} for the
rate sum is independent of $U_1$ and $U_2$, hence it also bounds
the single-letter region $\mathcal{R}_{BSL}(q)$.
\end{proof}
%%%%%%%%%%%%%%%%%%%%%%%%%%%%%%%%%%%%%55%%%%%%%%%%%%%%%%%%%%%%%%%%%%%%%%%%%%%%%%%%%%%%

A simplified expression for the function $F_{max}(q_1,q_2)$ of
\eqref{eq:FmaxDefinition} is shown in the following lemma.
\begin{lemma}\label{lem:EquivalentBSL1}
The function $F_{max}(q_1,q_2)$ \eqref{eq:FmaxDefinition} is given
by
\begin{align}
&F_{max}(q_1,q_2) =
\max_{\alpha_1,\alpha_2\in[0,1/2]}\Big[H_b(\alpha_1)+H_b(\alpha_2)-H_b\Big([\alpha_1-q_1]^+
\ast [\alpha_2-q_2]^+\Big)-1\Big]^+,\label{eq:SumRateBound1}
\end{align}
where $\ast$ is the binary convolution, i.e., $x \ast y \triangleq
(1-x)y+(1-y)x$.
\end{lemma}
\begin{proof}
The function $F_{max}(q_1,q_2)$ is defined in
\eqref{eq:FDefinition} and \eqref{eq:FmaxDefinition} where
$V_1,V'_1,V_2,V'_2$ are binary random variables. Let us define the
following probabilities:
\begin{align*}
\alpha_i&\triangleq P(V_i=1)\\
\delta_i&\triangleq P(V'_i=1|V_i=0)\\
\gamma_i&\triangleq P(V'_i=0|V_i=1),
\end{align*}
for $i=1,2$. We thus have
\begin{align*}
P(V'_i=1)&=(1-\alpha_i)\delta_i+\alpha_i(1-\gamma_i)\triangleq g(\alpha_i,\delta_i\gamma_i)\\
P(V_i \neq V'_i)&=\alpha_i\gamma_i+(1-\alpha_i)\delta_i\triangleq
h(\alpha_i,\delta_i,\gamma_i),
\end{align*}
for $i=1,2$. The maximization \eqref{eq:FmaxDefinition} can be
written as
\begin{align}\label{eq:MaxMinDef}
&F_{max}(q_1,q_2) =
\max_{\alpha_1,\alpha_2}\Big[H_b(\alpha_1)+H_b(\alpha_2)-\min_{\stackrel{\scriptstyle
\gamma_1,\delta_1,\gamma_2,\delta_2}{h(\alpha_i,\delta_i,\gamma_i)\leq
q_i,\;i=1,2}}H_b\Big(g(\alpha_1,\delta_1,\gamma_1) \ast
g(\alpha_2,\delta_2,\gamma_2)\Big)-1\Big]^+.
\end{align}
This maximization has two equivalent solutions
$(\alpha_1^o,\alpha_2^o)$ and $(1-\alpha_1^o,1-\alpha_2^o)$ where
$0\leq \alpha_1^o,\alpha_2^o \leq 0.5$, since any other
$(\alpha_1,\alpha_2)$ can only increase the inner minimization in
\eqref{eq:MaxMinDef} which results in a lower $F_{max}(q_1,q_2)$.
Therefore, without loss of generality we may assume that
$0\leq\alpha_1,\alpha_2\leq 0.5$.

To prove the lemma we need to show that for any $\alpha_i$ the
inner minimization is achieved by
\begin{align*}
\delta_i=0,\gamma_i=\min\{1,q_i/\alpha_i\},\;i=1,2.
\end{align*}
In other words, $V'_i$ has the smallest possible probability for
$1$ under the constraint that $P(V_i\neq V'_i)\leq q_i$, implying
that the transition from $V_i$ to $V'_i$ is a ``Z channel''. The
inner minimization requires that $P(V'_i=1)$ will be minimized
restricted to the constraint $P(V_i \neq V'_i)\leq q_i$, therefore
it is equivalent to the following minimization
\begin{align*}
\min_{\stackrel{\scriptstyle
\gamma_i,\delta_i}{h(\alpha_i,\delta_i\gamma_i)\leq
q_i}}g(\alpha_i,\delta_i\gamma_i),\;i=1,2.
\end{align*}
For $\alpha_i\leq q$, the solution is $\delta_i=0$ and
$\gamma_i=1$ since in this case $g(\alpha_i,\gamma_i,\delta_i)=0$
and the constraint is satisfied. For $q\leq\alpha_i\leq 0.5$, in
order to minimize $g(\alpha_i,\gamma_i,\delta_i)$, it is required
that $\delta_i\in[0,q/(1-a_i)]$ will be minimal and
$\gamma_i\in[0,q/\alpha_i]$ will be maximal such that the
constraint is satisfied. Clearly, the best choice is for
$\delta_i=0$ and $\gamma_i=q/\alpha_i$, in this case the
constraint is satisfies and
$g(\alpha_i,\gamma_i,\delta_i)=\alpha_i-q$.
\end{proof}

%%%%%%%%%%%%%%%%%%%%%%%%%%%%%%%%%%%%%55%%%%%%%%%%%%%%%%%%%%%%%%%%%%%%%%%%%%%%%%%%%%%%
The next lemma gives an explicit upper bound for
$F_{max}(q_1,q_2)$ \eqref{eq:FmaxDefinition} for the case that
$q_1=q_2$. Let
\begin{align}\label{eq:fxDefinition}
f(x)= x-\frac{1}{1+\Big(\frac{1}{x}-1\Big)^2},
\end{align}
and let
\begin{align}\label{eq:qcDefinition}
q_c\triangleq\max_{x\in[0,1/2]}f(x).
\end{align}
Since $f(x)$ is differentiable, we can characterize $q_c$ by
differentiating $f(x)$ with respect to $x$ and equating to zero,
thus we get that
\begin{align*}
4x^4-8x^3+10x^2-6x+1=0.
\end{align*}
This fourth order polynomial has two complex roots and two real
roots, where one of its real roots is a local minimum and the
other root is a local maximum. Specifically, this local maximum
maximizes $f(x)$ for the interval $x\in[0,1/2]$ and it achieves
$q_c\simeq 0.1501$ which occurs at $x\simeq 0.257$.
\begin{lemma}\label{lem:OuterBound}
For $q_1=q_2=q$, we have that:
\begin{align}\label{eq:FmaxOuterBound}
\begin{array}{ll}
  F_{max}(q,q)= 2H_b(q)-1, & q_c \leq q \leq 1/2\\
  F_{max}(q,q) \leq C^*q, & 0 < q < q_c \\
  F_{max}(0,0) = 0, & q=0, \\
\end{array}
\end{align}
where $q_c$ is defined in \eqref{eq:qcDefinition}, while
$C^*=\frac{2H_b(q^*)-1}{q^*}$ and $q^* \triangleq
1-1/\sqrt{2}\simeq 0.3$ are defined in \eqref{eq:SumCapacityBSL}.
\end{lemma}
Note that in the first case ($q_c \leq q \leq 1/2$) in
\eqref{eq:SumRateBound1} is achieved by $\alpha_1=\alpha_2=q$,
while in the third case ($q=0$) \eqref{eq:SumRateBound1} is
achieved by $\alpha_1=\alpha_2=1/2$ as shown in
\figref{AppendixFig3}. Although, we do not have an explicit
expression for $F_{max}(q,q)$ in the range $0 < q < q_c$, the
bound $F_{max}(q,q) \leq C^*q$ is sufficient for the purpose of
proving \thrmref{BSLRate} because $q_c\leq q^*$. In
\figref{AppendixFig2} a numerical characterization of
$F_{max}(q,q)$ is plotted.

\begin{figure}[h]
\begin{center}
\epsfig{file=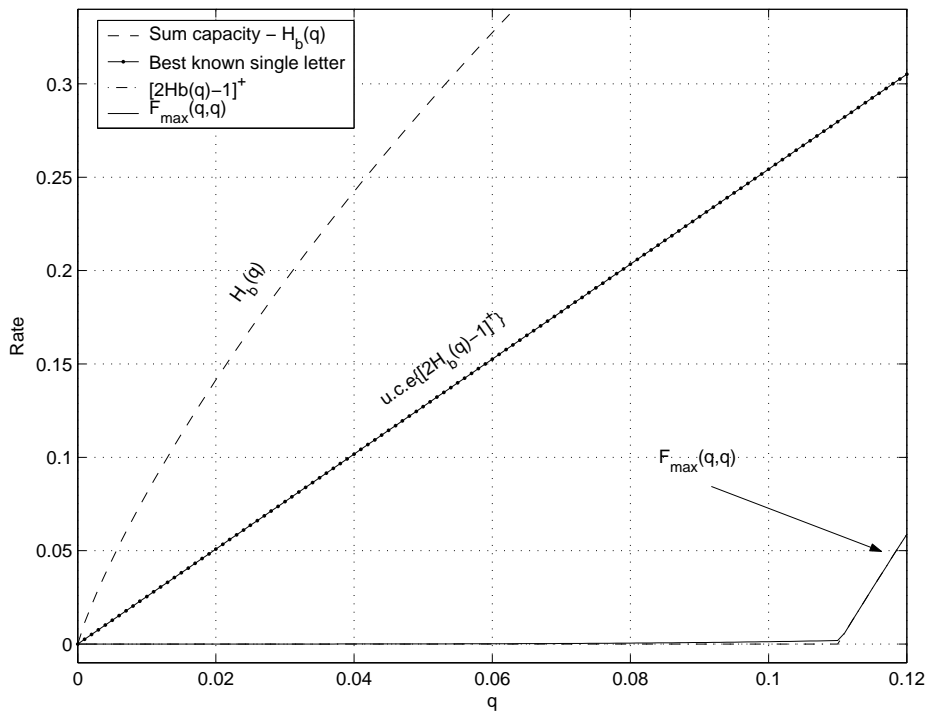,scale=0.7} \caption{Numerical
results of $F_{max}(q,q)$ \eqref{eq:SumRateBound1} for
$q\in[0,0.12]$ (\figref{BinaryDMAC} is the same plot for
$q\in[0,0.5]$) .}\label{fig:AppendixFig2}
\end{center}
\end{figure}

\begin{figure}[h]
\begin{center}
\epsfig{file=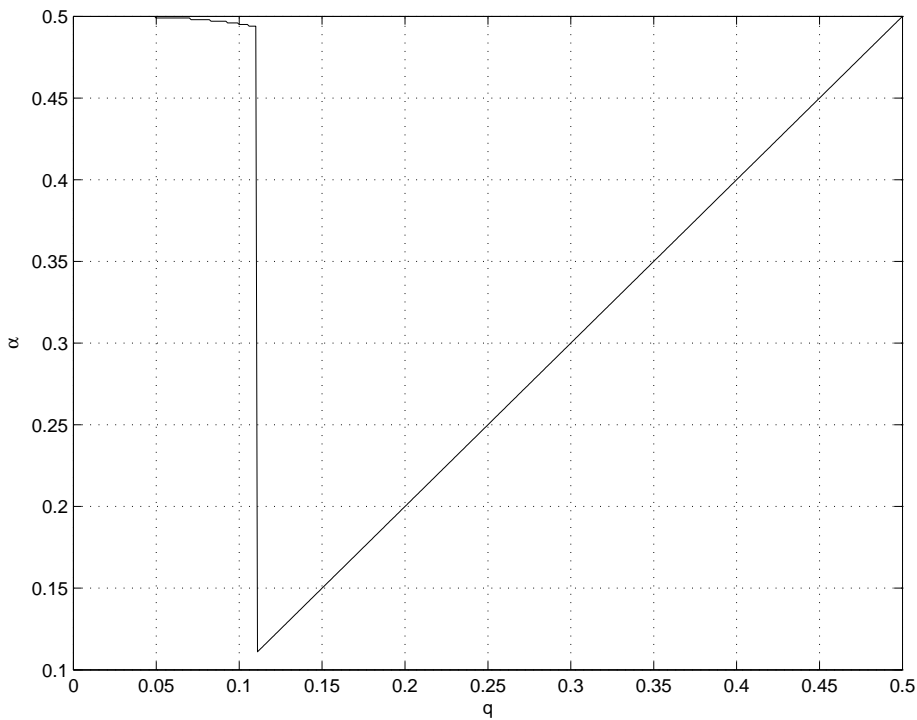,scale=0.7} \caption{The optimal
$\alpha_1=\alpha_2=\alpha(q)$ which maximizes
\eqref{eq:SumRateBound1}.}\label{fig:AppendixFig3}
\end{center}
\end{figure}

\begin{proof}
Define
\begin{align*}
F(\alpha_1,\alpha_2,q)\triangleq
H_b(\alpha_1)+H_b(\alpha_2)-H_b\big([\alpha_1-q]^+ \ast
[\alpha_2-q]^+\big)-1.
\end{align*}
From the discussion above about the cases of equality in
\eqref{eq:FmaxOuterBound}, \lemref{OuterBound} will follow by
showing that $F(\alpha_1,\alpha_2,q)$ is otherwise smaller, i.e.,
\begin{align}\label{eq:750}
F(\alpha_1,\alpha_2,q)\leq \Bigg\{
\begin{array}{cc}
  C^*q, & 0\leq q \leq q_c \\
  2H_b(q)-1, & q_c \leq q \leq 1/2\\
\end{array}
\end{align}
for all $0\leq \alpha_1,\alpha_2\leq 1/2$. It is easy to see that
for $\alpha_1,\alpha_2 \leq q$ the function
$F(\alpha_1,\alpha_2,q)$ is monotonically increasing with
$\alpha_1,\alpha_2$, and thus $F(\alpha_1,\alpha_2,q)\leq
F(q,q,q)=2H_b(q)-1$. For $\alpha_1 \leq q$ and $q < \alpha_2 \leq
1/2$, $F(\alpha_1,\alpha_2,q)$ is increasing with $\alpha_1$ and
decreasing with $\alpha_2$, and thus $F(\alpha_1,\alpha_2,q)\leq
F(q,q,q)=2H_b(q)-1$. Clearly, from symmetry, also for $\alpha_2
\leq q$ and $q \leq \alpha_1 \leq 1/2$,
$F(\alpha_1,\alpha_2,q)\leq 2H_b(q)-1$. As a consequence, we have
to show that \eqref{eq:750} is satisfied only for $q \leq
\alpha_1,\alpha_2 \leq 1/2$. Likewise, in the sequel we may assume
without loss of generality that $q \leq \alpha_2 \leq \alpha_1
\leq 1/2$.

\textbf{\textit{The bound for the interval $q_c < q \leq 1/2$}}:
in this case \eqref{eq:750} is equivalent to the following bound
\begin{align}\label{eq:800}
H_b\Big((\alpha_1-q) \ast (\alpha_2-q)\Big) - H_b(\alpha_1) -
H_b(\alpha_2) + 2H_b(q) \geq 0,\;\mbox{for}\;q_c \leq q \leq
\alpha_2 \leq \alpha_1\leq 1/2.
\end{align}
The LHS is lower bounded by
\begin{align}
&H_b\Big((\alpha_1-q) \ast (\alpha_2-q)\Big) - H_b(\alpha_1) - H_b(\alpha_2) + 2H_b(q)\nonumber\\
& \geq H_b(\alpha_1-q) - H_b(\alpha_1) - H_b(\alpha_2) + 2H_b(q)\label{eq:700}\\
& \geq H_b(\alpha_1-q) - 2H_b(\alpha_1) + 2H_b(q)\label{eq:710}\\
& \geq 0,\label{eq:715}
\end{align}
where \eqref{eq:700} follows since $H_b\Big((\alpha_1-q) \ast
(\alpha_2-q)\Big) \geq H_b(\alpha_1-q)$; \eqref{eq:710} follows
since $\alpha_2 \leq \alpha_1 \leq 1/2$; \eqref{eq:715} follows
from \lemref{f2} below.

\textbf{\textit{The bound for the interval $0 \leq q \leq q_c$}}:
in this case \eqref{eq:750} is equivalent to the following bound
\begin{align}
H_b\Big((\alpha_1-q) \ast (\alpha_2-q)\Big) \geq
H_b(\alpha_1)+H_b(\alpha_2)-1-C^*\cdot q,\;for\;0 \leq q \leq
\alpha_2 \leq \alpha_1 \leq q_c.\label{eq:LinearBound}
\end{align}
For fixed $\alpha_1$ and $\alpha_2$, let us denote the RHS and the
LHS of \eqref{eq:LinearBound} as
\begin{align*}
g_l(q) &\triangleq H_b\Big((\alpha_1-q) \ast (\alpha_2-q)\Big)\\
g_r(q) &\triangleq H_b(\alpha_1)+H_b(\alpha_2)-1-C^*\cdot q.
\end{align*}
The function $g_l(q)$ is convex-$\cap$ in $q$, since it is a
composition of the function $H_b(x)$ which is non-decreasing
convex-$\cap$ in the range $[0,1/2]$ and the function
$[\alpha_1-q] \ast [\alpha_2-q]$ which is convex-$\cap$ in $q$
\cite{BoydBook}. Since $g_r(q)$ is linear function in $q$ and
$g_l(q)$ is convex-$\cap$ function in $q$, the bound
\eqref{eq:LinearBound} is satisfied if the interval edges ($q=0$
and $q=\alpha_2$) satisfy this bound. For $q=0$,
\eqref{eq:LinearBound} holds since
\begin{align*}
g_l(q=0) &= H_b(\alpha_1 \ast \alpha_2 )\\
&\geq \max\{H_b(\alpha_1),H_b(\alpha_2)\}\\
&\geq \min\{H_b(\alpha_1),H_b(\alpha_2)\}\\
&\geq H_b(\alpha_1)+H_b(\alpha_2)-1\\
&= g_r(q=0).
\end{align*}
For $q=\alpha_2$ where $0 \leq q \leq q_c$, the bound
\eqref{eq:LinearBound} is satisfied since
\begin{align}
g_r(q=\alpha_2) &=
H_b(\alpha_1)+H_b(\alpha_2)-1-C^*\cdot\alpha_2\label{eq:900}\\
&\leq H_b(\alpha_1)-H_b(q^*)+H_b(0.5q^*)-0.5\label{eq:910}\\
&\leq H_b(\alpha_1)-H_b(q_c)\label{eq:920}\\
&\leq H_b(\alpha_1)-H_b(\alpha_2)\label{eq:930}\\
&\leq H_b(\alpha_1-\alpha_2)\label{eq:940}\\
&=g_l(q=\alpha_2),
\end{align}
where \eqref{eq:910} follows from \lemref{gr} since
$\arg\max_{\alpha_2\in[0,1/2]} g_r(\alpha_2)=0.5q^*$, and since
$C^*=\frac{2H_b(q^*)-1}{q^*}$;\eqref{eq:920} follows since for
$q^*=1-1/\sqrt{2}$ and $q_c$ defined in \eqref{eq:qcDefinition},
we have
$H_b\big(1-1/\sqrt{2}\big)-H_b\big(0.5(1-1/\sqrt{2})\big)+0.5\simeq
0.68...\geq H_b(q_c)$; \eqref{eq:930} follows since $q_c\geq
\alpha_2$, thus $H_b(q_c)\geq H_b(\alpha_2)$; \eqref{eq:940}
follows since $H_b(\alpha_1)-H_b(\alpha_1-\alpha_2)$ is decreing
in $\alpha_1$, thus $H_b(\alpha_1)-H_b(\alpha_1-\alpha_2)\leq
H_b(\alpha_2)$ for $\alpha_2\leq\alpha_1\leq 1/2$. Therefore, the
bound \eqref{eq:LinearBound} follows which completes the proof.
\end{proof}

%%%%%%%%%%%%%%%%%%%%%%%%%%%%%%%%%%%%%55%%%%%%%%%%%%%%%%%%%%%%%%%%%%%%%%%%%%%%%%%%%%%%
\lemref{f2} and \lemref{gr} are auxiliary lemmas used in the proof
of \lemref{OuterBound}.

%%%%%%%%%%%%%%%%%%%%%%%%%%%%%%%%%%%%%55%%%%%%%%%%%%%%%%%%%%%%%%%%%%%%%%%%%%%%%%%%%%%%
\begin{lemma}\label{lem:f2}
For $q_c \leq q \leq \alpha_1 \leq 1/2$, the following inequality
is satisfied
\begin{align}
f_1(\alpha_1)\triangleq H_b(\alpha_1-q) - 2H_b(\alpha_1) + 2H_b(q)
\geq 0.
\end{align}
\end{lemma}
\begin{proof}
Since $f_1(\alpha_1=q)=0$, it is sufficient to show that
$f_1(\alpha_1)$ is non-decreasing function in $\alpha_1$, i.e.,
$\frac{d}{d\alpha_1}f_1(\alpha_1)\geq 0$ for $q_c\leq q \leq
\alpha_1 \leq 1/2$, therefore
\begin{align}
\frac{d}{d\alpha_1}f_1(\alpha_1)=\log_2\Big(\frac{1}{\alpha_1-q}-1\Big)-2\log_2\Big(\frac{1}{\alpha_1}-1\Big)
\geq 0.\label{eq:720}
\end{align}
Due to monotonicity of the log function \eqref{eq:720} is
equivalent to
\begin{align}\label{eq:fCondition}
q \geq
\alpha_1-\frac{1}{1+\Big(\frac{1}{\alpha_1}-1\Big)^2}=f(\alpha_1),
\end{align}
where $f(\cdot)$ was defined in \eqref{eq:fxDefinition}. Since by
the definition of $q_c$ \eqref{eq:qcDefinition} $f(x)\leq
q_c\;\forall x\in[0,1/2]$, it follows that $f(\alpha_1)\leq
q\;\forall\;\alpha_1$ if $q_c\leq q$, and in particular for
$q_c\leq q \leq\alpha_1$, which implies \eqref{eq:fCondition} as
desired.
\end{proof}

%%%%%%%%%%%%%%%%%%%%%%%%%%%%%%%%%%%%%55%%%%%%%%%%%%%%%%%%%%%%%%%%%%%%%%%%%%%%%%%%%%%%

\begin{lemma}\label{lem:gr}
Let
\begin{align}
f_2(x)= H_b(x)-1-C^*\cdot x,
\end{align}
where $x\in[0,1/2]$, and $C^*=\frac{2H_b(q^*)-1}{q^*}$ where
$q^*=1-1/\sqrt{2}$. The maximum of $f_2(x)$ is achieved by
\begin{align}
\arg\max_{x}f_2(x)=0.5q^*=\frac{1}{2}(1-1/\sqrt{2}).
\end{align}
\end{lemma}
\begin{proof}
By differentiating $f_2(x)$ with respect to $x$ and comparing to
zero, we get that
\begin{align}
0=\frac{d}{dx}f_2(x)=\log_2\Big(\frac{1-x}{x}\Big)-C^*,
\end{align}
thus $x^o=\frac{1}{2^{C^*}+1}$ maximizes $f_2(x)$ since the second
derivative is negative, i.e., $\frac{d^2}{x^2}f_2(x)|_{x=x^o} <
0$. The lemma is followed since $x^o=\frac{1}{2^{C^*}+1}=0.5q^*$.
\end{proof}
%%%%%%%%%%%%%%%%%%%%%%%%%%%%%%%%%%%%%55%%%%%%%%%%%%%%%%%%%%%%%%%%%%%%%%%%%%%%%%%%%%%%

We are now in a position to summarize the proof of
\thrmref{BSLRate}.

\noindent \textbf{Proof of Theorem \ref{thm:BSLRate} - Converse
Part.}
The rate sum is upper bounded by
\begin{align}
R_1+R_2 &\leq u.c.e\Big\{F_{max}(q,q)\Big\}\label{eq:1100}\\
&\leq u.c.e\Bigg\{
\begin{array}{cc}
C^*\cdot q, & 0 \leq q \leq q_c \\
2H_b(q)-1, & q_c < q \leq 1/2 \\
\end{array}\Bigg\}\label{eq:1110}\\
&= u.c.e\Big\{[2H_b(q)-1]^+\Big\},\label{eq:1130}
\end{align}
where \eqref{eq:1100} follows from \lemref{EquivalentBSL};
\eqref{eq:1110} follows from \lemref{OuterBound}; and
\eqref{eq:1130} follows since \eqref{eq:1110} is equal to the
upper convex envelope of $[2Hb(q)-1]^+$.

%%%%%%%%%%%%%%%%%%%%%%%%%%%%%%%%%%%%%55%%%%%%%%%%%%%%%%%%%%%%%%%%%%%%%%%%%%%%%%%%%%%%
\section{A simplified Outer Bound for the Sum Capacity in the Strong Interference Gaussian Case}\label{App:Appendix3}
\begin{lemma}\label{lem:GaussianBSLRate}
The best known single-letter sum capacity
\eqref{eq:SingleLetterRegion} of the Gaussian doubly-dirty MAC
\eqref{eq:GaussianDirtyMAC} with power constraints $P_1$, $P_2$,
and strong interferences ($Q_1,Q_2\rightarrow\infty$) is upper
bounded by
\begin{align}
&R_1+R_2\leq u.c.e
\bigg\{\sup_{V_1,V'_1,V_2,V'_2}\Big[h(V_1)+h(V_2)-h\big(V'_1+V'_2+Z\big)+h(S_1+S_2)-h(S_1)-h(S_2)\Big]^+\bigg\},
\end{align}
where $u.c.e$ is the upper convex envelope operation with respect
to $P_1$ and $P_2$, and $[x]^+ = \max(0,x)$. The supremum is over
all $V_1,V'_1,V_2,V'_2$ such that $(V_1,V'_1)$ is independent of
$(V_2,V'_2)$, and
\begin{align*}
&E\Big\{(V_i-V'_i)^2\Big\} \leq P_i,\\
&h(V_i)\leq h(S_i),
\end{align*}
for $i=1,2$.
\end{lemma}
\begin{proof}
Let us define the following functions (corresponds to
$F(P_{V_1,V'_1},P_{V_2,V'_2})$ of \eqref{eq:FDefinition})
\begin{align}
&G\big(f_{V_1,V'_1},f_{V_2,V'_2}\big)\triangleq
\Big[h(V_1)+h(V_2)-h\big(V'_1+V'_2+Z\big)+h(S_1+S_2)-h(S_1)-h(S_2)\Big]^+.\label{eq:GDefinition}
\end{align}
The second function is the following maximization of
\eqref{eq:GDefinition} with respect to $V_1,V'_1,V_2,V'_2$.
\begin{align}
G_{max}(P_1,P_2)\triangleq &\sup_{V_1,V'_1,V_2,V'_2}G\big(f_{V_1,V'_1},f_{V_2,V'_2}\big)\label{eq:GmaxDefinition}\\
&\mbox{s.t}\;\;E\Big\{(V_i-V'_i)^2\Big\}\leq P_i,\quad h(V_i)\leq
h(S_i),\quad\mbox{for}\;\;i=1,2.\nonumber
\end{align}
Finally, we define the upper convex envelope of $G_{max}(P_1,P_2)$
with respect to $P_1$ and $P_2$:
\begin{align}
&\overline{G}_{max}(P_1,P_2)\triangleq
u.c.e\Big\{G_{max}(P_1,P_2)\Big\}.\label{eq:GuceDefinition}
\end{align}

Clearly if we take only the rate sum equation in
\eqref{eq:SingleLetterRegion0} we get an outer bound on the best
known single-letter region,
\begin{align}
&R_{BSL}^{sum}(U_1,U_2)\triangleq\Big[I(U_1,U_2;Y)-I(U_1,U_2;S_1,S_2)\Big]^+\label{eq:1000}\\
&=\Big[h(S_1|U_1)+h(S_2|U_2)-h(Y|U_1,U_2)+h(Y)-h(S_1)-h(S_2)\Big]^+\label{eq:1010}\\
&\leq \Big[h(S_1|U_1)+h(S_2|U_2)-h(Y|U_1,U_2)+h(S_1+S_2)-h(S_1)-h(S_2)\Big]^++o(1)\label{eq:1020}\\
&=\Bigg[E_{U_1,U_2}\Big\{h(S_1|U_1=u_1)+h(S_2|U_2=u_2)-h(Y|U_1=u_1,U_2=u_2)+h(S_1+S_2)-h(S_1)-h(S_2)\Big\}\Bigg]^++o(1)\label{eq:1030}\\
&\leq E_{U_1,U_2}\Bigg\{\Big[h(S_1|U_1=u_1)+h(S_2|U_2=u_2)-h(X_1+S_1+X_2+S_2+Z|U_1=u_1,U_2=u_2)\nonumber\\
&\qquad\qquad+h(S_1+S_2)-h(S_1)-h(S_2)\Big]^+\Bigg\}+o(1)\label{eq:1035}\\
&= E_{U_1,U_2}\Bigg\{G\Big(f_{S_1,S_1+X_1|U_1=u_1}f_{S_2,S_2+X_2|U_2=u_2}\Big)\Bigg\}+o(1)\label{eq:1040}\\
&\leq E_{U_1,U_2}\Big\{\overline{G}_{max}\big(P_{1|u_1},P_{2|u_2}\big)\Big\}+o(1)\label{eq:1050}\\
&\leq \overline{G}_{max}\big(E_{U_1}P_{1|u_1},E_{U_2}P_{2|u_2}\big)+o(1)\label{eq:1060}\\
&\leq \overline{G}_{max}\big(P_1,P_2\big)+o(1),\label{eq:1070}
\end{align}
where \eqref{eq:1020} follows since $h(Y)\leq h(S_1+S_2)+o(1)$
where $o(1)\rightarrow 0$ as $Q_1,Q_2\rightarrow\infty$;
\eqref{eq:1030} follows from the definition of the conditional
entropy; \eqref{eq:1035} follows since $[Ex]^+\leq E\{x^+\}$ and
since $Y=X_1+S_1+X_2+S_2+Z$; \eqref{eq:1040} follows from the
definition of the function $G\big(f_{V_1,V'_1},f_{V_2,V'_2}\big)$
\eqref{eq:GDefinition}, likewise \eqref{eq:1050} follows from the
definition of the function $\overline{G}_{max}(P_1,P_2)$
\eqref{eq:GuceDefinition}, and since $h(S_i|U_i)\leq h(S_i)$ and
from the definition
\begin{align*}
P_{i|u_i}\triangleq E\Big\{X_i^2|U_i=u_i\Big\},\;for\;i=1,2;
\end{align*}
\eqref{eq:1060} follows from Jensen's inequality since
$\overline{G}_{max}(P_1,P_2)$ is a concave function;
\eqref{eq:1070} follows from the input constraints where
\begin{align}
EX_i^2&=E_{U_i}E\big\{X_i^2|U_i=u_i\big\}= E_{U_i}P_{i|u_i}\leq
P_i,\;\mbox{for}\;\;i=1,2.\label{eq:GInputConst}
\end{align}
The lemma follows since the upper bound \eqref{eq:1070} for the
rate sum is now independent of $U_1$ and $U_2$, hence it also
bound the single-letter region $\mathcal{R}_{BSL}(P_1,P_2)$.
\end{proof}

%%%%%%%%%%%%%%%%%%%%%%%%%%%%%%%%%%%%%55%%%%%%%%%%%%%%%%%%%%%%%%%%%%%%%%%%%%%%%%%%%%%%
\section*{Acknowledgment}
The authors wish to thank Ashish Khisti for earlier discussions on
the binary case. The authors also would like to thank Uri Erez for
helpful comments.

%%%%%%%%%%%%%%%%%%%%%%%%%%%%%%%%%%%%%55%%%%%%%%%%%%%%%%%%%%%%%%%%%%%%%%%%%%%%%%%%%%%%
\bibliographystyle{IEEEtran}
\bibliography{../../../Bibliography/mybib}

\begin{thebibliography}{10}
\providecommand{\url}[1]{#1}
\csname url@rmstyle\endcsname
\providecommand{\newblock}{\relax}
\providecommand{\bibinfo}[2]{#2}
\providecommand\BIBentrySTDinterwordspacing{\spaceskip=0pt\relax}
\providecommand\BIBentryALTinterwordstretchfactor{4}
\providecommand\BIBentryALTinterwordspacing{\spaceskip=\fontdimen2\font plus
\BIBentryALTinterwordstretchfactor\fontdimen3\font minus
  \fontdimen4\font\relax}
\providecommand\BIBforeignlanguage[2]{{%
\expandafter\ifx\csname l@#1\endcsname\relax
\typeout{** WARNING: IEEEtran.bst: No hyphenation pattern has been}%
\typeout{** loaded for the language `#1'. Using the pattern for}%
\typeout{** the default language instead.}%
\else
\language=\csname l@#1\endcsname
\fi
#2}}

\bibitem{Costa83}
M.~Costa, ``Writing on dirty paper,'' \emph{IEEE Trans. Information Theory},
  vol. IT-29, pp. 439--441, May 1983.

\bibitem{GelfandPinsker80}
S.~Gelfand and M.~S. Pinsker, ``Coding for channel with random parameters,''
  \emph{Problemy Pered. Inform. (Problems of Inform. Trans.)}, vol. 9, No. 1,
  pp. 19--31, 1980.

\bibitem{CoverBook}
T.~M. Cover and J.~A. Thomas, \emph{Elements of Information Theory}.\hskip 1em
  plus 0.5em minus 0.4em\relax New York: Wiley, 1991.

\bibitem{BaruchShamaiVerdu06}
A.~Somekh-Baruch, S.~Shamai, and S.~Verdu, ``Cooperative encoding with
  asymmetric state information at the transmitters,'' in \emph{Proceedings 44th
  Annual Allerton Conference on Communication, Control, and Computing, Univ. of
  Illinois, Urbana, IL, USA}, Sep. 2006.

\bibitem{KotagiriLaneman06}
S.~Kotagiri and J.~N. Laneman, ``Multiple access channels with state
  information known at some encoders,'' \emph{IEEE Trans. Information Theory},
  July 2006, submitted for publication.

\bibitem{Jafar06}
S.~A. Jafar, ``Capacity with causal and non-causal side information - a unified
  view,'' \emph{IEEE Trans. Information Theory}, vol. IT-52, pp. 5468--5475,
  Dec. 2006.

\bibitem{Marton79}
K.~Marton, ``A coding theorem for the discrete memoryless broadcast channel,''
  \emph{IEEE Trans. Information Theory}, vol. IT--22, pp. 374--377, May 1979.

\bibitem{ErezShamaiZamir05}
U.~Erez, S.~Shamai, and R.~Zamir, ``Capacity and lattice strategies for
  canceling known interference,'' \emph{IEEE Trans. Information Theory}, vol.
  IT-51, pp. 3820--3833, Nov. 2005.

\bibitem{ZamirShamaiEreznested}
R.~Zamir, S.~Shamai, and U.~Erez, ``Nested linear/lattice codes for structured
  multiterminal binning,'' \emph{IEEE Trans. Information Theory}, vol. IT-48,
  pp. 1250--1276, June 2002.

\bibitem{PhilosofKhistiErezZamir07}
T.~Philosof, A.~Khisti, U.~Erez, and R.~Zamir, ``Lattice strategies for the
  dirty multiple access channel,'' in \emph{Proceedings of IEEE International
  Symposium on Information Theory, Nice, France}, June 2007.

\bibitem{KornerMarton79}
J.~Korner and K.~Marton, ``How to encode the modulo-two sum of binary
  sources,'' \emph{IEEE Trans. Information Theory}, vol. IT-25, pp. 219--221,
  March 1979.

\bibitem{GopinathBook}
T.~M. Cover and B.~Gopinath, \emph{Open Problems in Communication and
  Computation}.\hskip 1em plus 0.5em minus 0.4em\relax New York:
  Springer-Verlag, 1987.

\bibitem{CsiszarBook}
I.~Csiszar and J.~Korner, \emph{Information Theory - Coding Theorems for
  Discrete Memoryless Systems}.\hskip 1em plus 0.5em minus 0.4em\relax New
  York: Academic Press, 1981.

\bibitem{NazerGastpar07}
B.~Nazer and M.~Gastpar, ``Computation over multiple-access channels,''
  \emph{IEEE Trans. Information Theory}, vol. IT-53, pp. 3498--3516, Oct. 2007.

\bibitem{KrithivasanPradhan07}
D.~Krithivasan and S.~S. Pradhan, ``Lattices for distributed source coding:
  Jointly {G}aussian sources and reconstruction of a linear function,''
  \emph{arXiv:cs.IT/0707.3461V1}.

\bibitem{KhistiPr}
A.~Khisti, ``Private communication.''

\bibitem{Gallager68}
R.~G. Gallager, \emph{Information {T}heory and {R}eliable
  {C}ommunication}.\hskip 1em plus 0.5em minus 0.4em\relax New York, N.Y.:
  Wiley, 1968.

\bibitem{CoveringBook}
G.~Cohen, I.~Honkala, S.~Litsyn, and A.~Lobstein, \emph{Covering Codes}.\hskip
  1em plus 0.5em minus 0.4em\relax Amsterdam, The Netherlands: North Holland
  Publishing, 1997.

\bibitem{AhlswedeKorner75}
R.~Ahlswede and J.~Korner, ``Source coding with side information and a converse
  for the degraded broadcast channel,'' \emph{IEEE Trans. Information Theory},
  vol.~21, pp. 629--637, 1975.

\bibitem{Wyner75}
A.~Wyner, ``On source coding with side information at the decoder,'' \emph{IEEE
  Trans. Information Theory}, vol. IT-21, pp. 294--300, 1975.

\bibitem{Berger77}
T.~Berger, \emph{Multiterminal Source Coding}.\hskip 1em plus 0.5em minus
  0.4em\relax New York: In {G.L}ongo, editor, the {I}nformation {T}heory
  {A}pproach to {C}ommunications, Springer-Verlag, 1977.

\bibitem{Forney89a}
L.~F. Wei and G.~D. Forney, ``Multidimensional constellation - part {I}:
  Introduction, figures of merit, and generalized cross constellations,''
  vol.~7, pp. 877--892, Aug. 1989.

\bibitem{CohenZamir07}
A.~Cohen and R.~Zamir, ``Entropy amplification property and the loss for
  writing on dirty paper,'' \emph{IEEE Trans. Information Theory}, To appear,
  April 2008.

\bibitem{BoydBook}
S.~Boyd and L.~Vandenberghe, \emph{Convex Optimization}.\hskip 1em plus 0.5em
  minus 0.4em\relax Cambridge: Cambridge University Press, 2004.

\end{thebibliography}
%\bibliography{mybib}
\end{document}